\documentclass[11pt]{article}
\usepackage[utf8]{inputenc}
\usepackage{amsfonts, amsmath, amssymb, amsthm}
\usepackage[english]{babel}
\usepackage[font=small,labelfont=bf]{caption}
\usepackage{listings}
\usepackage{subcaption}
\usepackage[left=1in, right=1in, top=1in, bottom=1in]{geometry}
\usepackage{verbatim}
\usepackage{titling}
\usepackage{parskip}
\usepackage{todonotes} 
\usepackage{graphicx} 
\usepackage{xcolor}
\usepackage{hyperref}
\usepackage{slashed}
\usepackage{empheq}
\usepackage{dsfont}
\usepackage{soul}
\usepackage{tikz}
\usetikzlibrary{arrows, decorations.markings,
  decorations.pathmorphing, shapes, fit, backgrounds, patterns,
  calc,petri,snakes, arrows.meta}
  \usetikzlibrary{positioning}
\usetikzlibrary{decorations.text}
\usetikzlibrary{decorations.pathmorphing}
\tikzset{>=latex}

\definecolor{arrowblue}{RGB}{98,145,224}

\usetikzlibrary{matrix,arrows,fit} 

\tikzset{circarrow/.style={
        *->,
        shorten <=-2pt
    }
}
\tikzstyle{place}=[circle,draw=blue!50,fill=blue!20,thick]
\tikzstyle{placespec}=[circle,draw=blue!50,fill=blue!20,thick, inner sep=0pt,minimum size=6.5mm]
\tikzstyle{place1}=[circle,draw=blue!50,fill=blue!20,thick,inner sep=0pt,minimum size=2mm]
\tikzstyle{place11}=[circle,draw=blue!50,fill=blue!20,thick, inner sep=0pt,minimum size=0mm]
\tikzstyle{placetab}=[circle,draw=blue!50,fill=blue!20,thick, inner sep=0pt,minimum size=2mm]
\tikzstyle{transition}=[rectangle,draw=blue!50,fill=blue!20,thick]
\tikzstyle{vecArrow} = [thick, decoration={markings,mark=at position
   1 with {\arrow[semithick]{open triangle 60}}},
   double distance=1.4pt, shorten >= 5.5pt,
   preaction = {decorate},
   postaction = {draw,line width=1.4pt, white,shorten >= 4.5pt}]

\tikzset{
BasicNode/.style={circle, draw= black!50, fill=colora!40, thin, minimum size
  = 5mm, inner sep = 0mm},
SusCirc/.style={BasicNode,fill=colorS!40},
InfCirc/.style={BasicNode,fill=colorI!40},
RecCirc/.style={BasicNode,fill=colorR!40},
DefCirc/.style={BasicNode}
}

\tikzset{
LargeDot/.style = {circle, draw = black, fill = black, thick, minimum
  size = 1mm, inner sep = 0mm}
}

\tikzstyle{BigCirc}=[circle,draw=black!50,thick,inner sep=0pt,minimum size=0pt]

\tikzstyle{default} = [draw, minimum size = 3em, text width = 4em, text centered]
\tikzstyle{wide}=[draw, minimum size=3em, text width=7.5em, text
centered]
\tikzstyle{shortbox}=[draw, minimum size=2em, text width=4.5em, text centered]
\tikzstyle{bigbox}=[draw, inner sep=20pt,label={[align=right,shift={(-1.5ex,3ex)}]south east:\llap{#1}}]
\tikzstyle{box} = [draw, minimum size=2em, text width=1.5em, text centered]

\tikzstyle{decay} = [draw, ->, decorate, decoration = {snake,
  amplitude = 0.4mm, segment length=2mm, post length = 2mm, pre length
= 1mm}]

\usepgflibrary{arrows.meta}
\usetikzlibrary{arrows.meta}

\tikzstyle{ClearCirc}=[circle,draw=black!50,fill=white!20,thick, inner sep=1pt]

\definecolor{colora}{RGB}{0,115,179}
\definecolor{colorb}{RGB}{230,154,0} 
\definecolor{colorc}{RGB}{0,154,128} 
\definecolor{colord}{RGB}{205,10,179}
\definecolor{colore}{RGB}{255,32,0}
\definecolor{colorf}{RGB}{240,228,66}
\definecolor{colorg}{RGB}{90,179,230}
\definecolor{colorh}{RGB}{205,154,179}

\definecolor{colorS}{RGB}{0,154,128}
\definecolor{colorI}{RGB}{255,32,0}
\definecolor{colorR}{RGB}{205,10,179}
\definecolor{colorE}{RGB}{240,228,66} 

\definecolor{peach}{HTML}{E69F00}
\definecolor{blue}{HTML}{56B4E9}
\definecolor{green}{HTML}{009E73}
\definecolor{yellow}{HTML}{F0E442}
\definecolor{blue2}{HTML}{0072B2}
\definecolor{grey}{HTML}{878787}
\definecolor{pink}{HTML}{CC79A7}

\usepackage{booktabs}
\usepackage[sort,numbers]{natbib}

\newcommand{\bu}{\ensuremath{\bullet}}
\newcommand{\ka}{\kappa}

\newtheorem{theorem}{Theorem}
\newtheorem{lemma}[theorem]{Lemma}

\title{Necessary and sufficient conditions for exact closures of epidemic equations on configuration model networks}
\author{Istv\'an Z. Kiss$^1$, Eben Kenah$^2$ \& Grzegorz A. Rempa\l{}a$^2$ }
\date{
$^1$ Department of Mathematics, University of Sussex, Falmer, Brighton BN1 9QH, UK\\
$^2$ Division of Biostatistics, College of Public Health and Mathematical Biosciences Institute, The Ohio State University, Columbus, OH, USA\\
	  \vspace{1cm}\today
	  }

\begin{document}
\maketitle
\begin{abstract}
We prove that the exact closure of SIR pairwise epidemic  equations on a configuration model network is  possible if and only if the degree distribution is Poisson, Binomial, or Negative binomial.
The proof relies on  establishing, for these specific degree distributions,  the equivalence of the closed pairwise model and the so-called  dynamical survival analysis (DSA) edge-based  model which was previously shown to be exact.
Indeed, as we show here, the DSA model is equivalent to the well-known edge-based Volz model.
We use this result to provide reductions of the closed pairwise and  Volz models to the same single equation involving only susceptibles, which has a useful statistical interpretation in terms of the times to infection.
We illustrate our findings  with some numerical examples.
%
\end{abstract}

Keywords: Epidemics, Networks, Inference, Pairwise models, Survival analysis.

\tableofcontents

\section{Introduction}


To understand the transmission dynamics of an infectious disease, one of the most important modelling components is the contact process between individuals in the population.
Changes in this contact structure often result in unexpected changes in epidemic dynamics, as observed in the COVID-19 pandemic as well as earlier epidemics~\cite{kresge2021analyzing,feld1991why}.
These observations have inspired a number of contact-based models of disease transmission (e.g., \cite{hosp22,kiss2017mathematics,gross2006epidemic,ball2019stochastic,risau2009contact,jacobsen2018large}), usually under various versions of   Susceptible-Infected-Recovered (SIR) dynamics.
In many of these models, contacts may be represented mathematically as a stochastic network (random graph) of $N$ individuals (nodes) formed  using the configuration model~\citep{molloy1995critical, bollobas2001random}, where  the node degrees are assumed to be independent and identically distributed as in the Newman–Strogatz–Watts (NSW) random graph construction~\cite{newman2001random}.  
Unfortunately, even for moderate values of $N$, the number of equations needed to describe the temporal dynamics of epidemics on such networks is usually too large to be tractable.
To circumvent this problem, some authors have used ``closure'' to create a reduced and closed system of equations by representing terms corresponding to larger structures in terms of smaller structures.
This representation involves an approximation in most cases, but it can be exact in the case of SIR dynamics on configuration model networks~\cite{kiss2017mathematics}.
	
In this paper, we consider the mean-field equations describing an SIR epidemic   on a configuration model network in terms of the evolving proportion of susceptibles, infecteds, and recovereds as $N$ tends to infinity.
In particular, we present necessary and sufficient conditions for the exact pairwise closure of a typical set of network equations describing the dynamics of $k$-tuples of nodes for $k \ge 1$~\cite{rand1999correlation}.
As it turns out, this condition requires that the degree distribution of the underlying configuration model graph be either Poisson, binomial, or negative binomial.
We also show that, once the condition is satisfied, the limiting pairwise model is equivalent to the edge-based model proposed by Volz~\cite{volz2008sir} and extended by Miller~\cite{miller2011note} as well as to the recently proposed network-based dynamical survival analysis (DSA) model~\cite{house2011insights,jacobsen2018large,khudabukhsh2020survival}. 
This latter equivalence is particularly useful from the viewpoint of statistical inference for the parameters of pairwise models, which have been applied to COVID-19 modelling~\cite{di2022dynamic}.  




 As was first shown by Volz~\cite{volz2008sir}, the limiting dynamics of a configuration model network with a specified degree distribution can be modelled with a system of three nonlinear ordinary differential equations.
 Volz's analysis used the probability generating function (PGF) of the degree distribution, as well as edge-centric quantities (such as the number of edges with nodes in certain states) rather than node-centric quantities (such as the number of infected or susceptible nodes).
 Soon afterwards, Decreusefond and colleagues~\cite{decreusefond2012large} used a more formal approach based on convergence of measure-valued processes to prove that Volz’s mean-field model is indeed  the large limit of a configuration model network.
This was followed by several law of large numbers (LLN) scaling limit results derived under varying sets of technical assumptions such as uniformly bounded degrees~\cite{bohman2012sir, barbour2013approximating}.
To date, the least restrictive set of assumptions seems to be that the degree of a randomly chosen node is uniformly integrable and the maximum degree of initially infected nodes is $o(N)$ where $N$ is the number of nodes~\cite{janson2014law}.

More recently, Jacobsen and colleagues~\cite{jacobsen2018large} provided an alternative method to derive the mean-field limit of a stochastic SIR model on a multi-layer network.
We  refer to it below as the dynamical survival analysis (DSA) model, due to its connection with statistical survival models. 
This approach also shows the exactness of Volz's  model in the large network limit and, perhaps more importantly, results in a mean-field model over variables different from those in Volz's approach.
The primary advantage of this alternative formulation is that it allows us to reinterpret the epidemic from the statistical survival analysis viewpoint by approximating  the probability that a typical node who was susceptible at time $t = 0$ is still susceptible at time $t > 0$.
Below, we show that Volz's model is equivalent to the DSA model and therefore admits the same survival equation. 

The rest of the paper is organised as follows.
In the next section, we provide some relevant background briefly describing stochastic dynamics on a configuration model network along with its different limiting approximations based on the pairwise, Volz, and DSA approaches.
In Section~3, we introduce and then characterise a class of \emph{Poisson-type} distributions and then present  our main result on the necessary and sufficient condition for the exact closure of the pairwise network model.
This result is more precise than that obtained in~\cite{jacobsen2018large}, but it is less general as it only covers single layer networks.
In Sections~4 and~5, we provide some additional details on DSA model's connection with statistical inference and offer concluding remarks.
Additional calculations on the DSA and Volz models are presented in the Appendix.  

\section{Network Epidemic Models}
We  describe  the underlying dynamics of the stochastic  SIR epidemic process on a network of size $N$ as follows:
At the start of an epidemic, we pick $m$ initially infectious individuals at random from the population.
An infectious individual remains so for an infectious period that is sampled from an exponential distribution with rate $\gamma$.
During this period, s/he makes contact with his or her immediate neighbours according to a  Poisson process with intensity $\beta$.
If the individual so contacted is still susceptible, s/he will immediately become infectious.
After the infectious period, the infectious individual recovers and is immune to further infections.
All infectious periods and Poisson processes are assumed to be independent of each other.
The epidemic is assumed to evolve on a configuration model network that is constructed as follows.
Each node is given a random number of half edges according to a specified degree distribution $(p_k)$, and all half edges are matched uniformly at random to form proper edges\footnote{
    Technically, this construction allows for self-loops and edges left unmatched,  but it turns out that such cases are rare and may be neglected as $N\to \infty$ (see, for instance,  \citep{Remko} Section 7.6, pp. 239).
}.
Although the exact behaviour of this SIR epidemic process is quite complicated, there exist several approximations that rely on aggregated or averaged quantities.
To describe them, generating functions are useful.

\subsection{Probability generating function }\label{ssec:pgf}
If $p_k$ is the probability that a randomly chosen node has degree $k$, then the probability generating function (PGF) of the degree distribution is:
\begin{equation}
    \psi(u) = \sum_{k = 0}^\infty p_k u^k
\end{equation}
The PGF $\psi$ contains a tremendous amount of information about epidemic dynamics on configuration model networks.
Let $\theta$ be the probability that an initially susceptible node of degree one remains uninfected at time $t$ in an infinite network.
Then, assuming no variation in infectiousness or susceptibility to infection among nodes except for their degree, the probability that a node with degree $k$ remains uninfected equals the probability $\theta^k$ that infection has not crossed any of its edges (see~\cite{volz2008sir}). 
Summing over all possible $k$ shows that $\psi(\theta)$ is the probability that a randomly chosen node remains susceptible in an infinite network.
The degree distribution of the remaining susceptible nodes has the PGF:
\begin{equation}
    u \mapsto 
    \frac{\sum_{k = 0}^\infty (p_k \theta^k) u^k}
        {\sum_{k = 0}^\infty p_k \theta^k}
    = \frac{\psi(\theta u)}{\psi(\theta)}
\end{equation}
which equals $\psi(u)$ when $\theta = 1$.
In general, $\psi$ tells us about the properties of a randomly chosen susceptible node.

The first derivative of $\psi$ tells us about the mean degree of susceptible nodes and about the properties of a node reached by crossing an edge.
At a given value of $\theta$, the mean degree of the remaining susceptible nodes is:
\begin{equation}
    \frac{\text{d}}{\text{d} u} \frac{\psi(\theta u)}{\psi(\theta)} \bigg|_{u = 1}
    = \frac{\theta \psi'(\theta)}{\psi(\theta)}
    \label{eq:meandeg}
\end{equation}
which equals $\psi'(1)$ when $\theta = 1$.
If we cross an edge, the probability that we end up at a node with degree $k$ is proportional to $k$.
If we start at a randomly chosen node and cross an edge, the number of edges we can use to reach a third node has the PGF:
\begin{equation}
    u \mapsto 
    \frac{\sum_{k = 1}^\infty (k p_k) u^{k - 1}}
        {\sum_{k = 1}^\infty k p_k}
    = \frac{\psi'(u)}{\psi'(1)}
    \label{eq:xdegpgf}
\end{equation}
If you are a susceptible node with a neighbour of degree $k$, this neighbour remains susceptible as long as infection has not crossed any of the $k - 1$ edges that lead to a third node.
At a given value of $\theta$, the probability that a neighbour of a susceptible node remains susceptible is $\psi'(\theta) / \psi'(1)$.
If we cross an edge to reach a susceptible neighbour, the number of edges we can cross to reach a third node has the PGF:
\begin{equation}
    u \mapsto 
    \frac{\sum_{k = 1}^\infty \big(k p_k \theta^k\big) u^{k - 1}}{\sum_{k = 1}^\infty k p_k \theta^k}
    = \frac{\psi'(\theta u)}{\psi'(\theta)}
\end{equation}
This degree distribution is often called the \emph{excess degree distribution}, and it plays an important role in the dynamics of epidemics on networks.

The second derivative of $\psi$ can be used to find the mean excess degree of susceptible nodes and the variance of their degree.
Given $\theta$, the mean excess degree of susceptible nodes is:
\begin{equation}
    \frac{\text{d}}{\text{d} u} \frac{\psi'(\theta u)}{\psi'(\theta)} \bigg|_{u = 1}
    = \frac{\theta \psi''(\theta)}{\psi'(\theta)}
    \label{eq:meanxdeg}
\end{equation}

\subsection{Pairwise model}
The  pairwise model  provides an intuitive way of  describing the dynamics of an SIR epidemic on a configuration model graph.
The pairwise model equations, as proposed for instance in~\cite{rand1999correlation},  are:
\begin{align}
    \begin{aligned}
        [\dot{S}]   &= -\beta [SI] \\
        [\dot{I}]   &= \beta [SI]-\gamma[I] \\
        [\dot{R}]   &= \gamma [I] \\
        [\dot{SI}]  &= -\gamma [SI] + \beta \big([SSI] - [ISI]\big) - \beta [SI] \\
        [\dot{SS}]  &= -2 \beta [SSI]
    \end{aligned}
    \label{eq:PW_unclosed}
\end{align}
where $[A]$, $[AB]$, $[ABC]$ with $A, B, C \in \{S, I, R\}$ stand for the number of singles, doubles and triples in the entire network with the given sequence of states when each group is counted in all possible ways.
More formally,
\begin{equation}
    [ABC] = \sum_{i = 1}^N \sum_{j = 1}^{N} \sum_{k=1}^N a_{ij} a_{jk} I_i(A) I_j(B) I_k(C)
    \label{eq:ABC}
\end{equation}
where $(a_{ij})_{i,j=1,2,\dots, N}$ is the adjacency matrix of the network with entries either zero or one and $I_i(A)$, $I_i(B)$, and $I_i(C)$ are binary variables that equal one when the status of $i$-th individual is $A$, $B$, and $C$, respectively, and equal zero otherwise.
The singles $[A]$ and doubles $[AB]$ are similarly defined.
We consider undirected networks with no self-loops, so $a_{ii}=0$ and $a_{ij}=a_{ji}$.

To completely describe the model, additional equations for triples are needed. These will depend on quadruples, which will depend on quintuples, and so on.
To make the model tractable in the face of an ever-increasing number of variables and equations, one often introduces the notion of a ``closure'' in which larger structures (e.g., triples) are represented by smaller ones (e.g. pairs).     
The  model~\eqref{eq:PW_unclosed} can be closed using the methods described in Section~\ref{sec:closure}.

The two models that we describe next do not require closure
and  are known to be exact in the large network limit (i.e., as $N \to \infty$)~\cite{decreusefond2012large, bohman2012sir, barbour2013approximating, janson2014law}. However, they are less straightforward  to interpret.


\subsection{Volz's model}

In addition to the limiting ($N\to \infty$) probability $\theta$ defined in Section~\ref{ssec:pgf},  let us also introduce the limiting  probabilities  $p_I$ and $p_S$  that  a  randomly selected  edge with one susceptible vertex is of type  $SI$ and $SS$, respectively. 
In this notation, Volz's mean-field equations~\cite{volz2008sir} are: 
\begin{equation}
    \begin{aligned}
        \dot{\theta} &=-\beta p_{I} \theta \\
        \dot{p}_{I} &=\beta p_{S} p_{I} \theta \frac{\psi''(\theta)}{\psi'(\theta)} - \beta p_{I} \left(1 - p_{I}\right) - \gamma p_{I} \\
        \dot{p}_{S} &=\beta p_{S} p_{I} \left(1 - \theta \frac{\psi''(\theta)}{\psi'(\theta)}\right) \\
        x_S &= \psi(\theta) \\
        \dot{x}_I &= \beta p_{I} \theta \psi'(\theta) - \gamma x_I
    \end{aligned}
    \label{eq:VOrig}
\end{equation}
where the derivative with respect to time is marked  with a dot and the derivative with respect to $\theta$ is marked with a prime.
Here $x_S$ and $x_I$ denote the limiting proportions of susceptibles and infecteds, respectively.
Note that the first three equations are decoupled from the remaining two and that the proportion of recovered may be obtained from the conservation relationship.
The initial conditions are 
\begin{equation}
    \begin{aligned}
        x_S(0) = \theta(0) = p_S(0) &= 1\\
        x_I(0) = p_I(0) &= \rho
    \end{aligned}
    \label{eq:Vol_ic}
\end{equation}
where $0 < \rho \ll 1$.

\subsection{DSA model}
An alternative description of the limiting  dynamics of a large  configuration model network under an SIR epidemic was given in~\cite{jacobsen2018large}. 
Although originally considered in the context of multi-layer networks, its single layer version has been applied recently to statistical inference problems under the name of dynamical survival analysis (DSA) model. 
In this approach, the limiting equations are derived in terms of the asymptotic proportions $x_{SI}$ of $SI$-type  and $x_{SS}$ of $SS$-type edges and the additional probability $x_\theta$.
In Appendix~B, we show that the latter coincides with  the probability $\theta$ defined in Section~\ref{ssec:pgf}.
The equations are:
\begin{equation}
    \begin{aligned}
        \dot{x}_\theta &= -\beta \frac{x_{SI}}{ \psi'\left(x_\theta\right)} \\
        \dot{x}_{SS} &=  -2 \beta x_{S I} x_{S S} \frac{ \psi''\left(x_\theta\right)}{ \psi'\left(x_\theta\right)^{2}} \\
        \dot{x}_{SI} &= x_{S I}\Bigg[\beta\big(x_{S S}-x_{S I}\big) \frac{ \psi''\left(x_\theta\right)}{ \psi'\left(x_\theta\right)^{2}} - (\beta + \gamma)\Bigg] \\
        \dot{x}_{S} &= -\beta x_{S I} \\
        \dot{x}_{I} &=  \beta x_{S I}-\gamma x_{I}
    \end{aligned}
    \label{eq:GR}
\end{equation}
As in Volz's system, the first three equations do not depend explicitly on the dynamics of $x_S$ and $x_I$ and therefore may be decoupled from the remaining two equations.
The initial conditions are:
\begin{equation}
    \begin{aligned}
        x_{S}(0) = x_{\theta}(0) &= 1 \\
        x_{I}(0) &= \rho \\
        x_{SS}(0) &= \mu \\
        x_{SI}(0) &= \mu \rho
    \end{aligned}
    \label{eq:GR_ic}
\end{equation}
where $0 < \rho \ll 1$ and $\mu > 0$

\section{Closing the pairwise model}
\label{sec:closure}
In practice, one  needs to define the time dynamics of the triples $[SSI]$ and $[ISI]$ to use the system~\eqref{eq:PW_unclosed}.
Typically, these equations are closed by approximating the dynamics of triples using pairs.
This method is referred to  as the ``pair approximation'' or ``pairwise closure'' in~\cite{jacobsen2018large}.

\subsection{Exact closure condition}
While various justifications of closures have been proposed before ~\cite{kiss2009contact}, we present a slightly different justification that is focused on the PGF of the degree distribution.
The form of the degree distribution plays a key role in obtaining necessary and sufficient conditions on the network to ensure that pairwise closures are exact.

Let $[A_j B_k C_l]$ indicate the number of connected triplets $ABC$ as defined in equation~\eqref{eq:ABC} such that the node in state $A$ has degree $j$, the node in state $B$ has degree $k$, and the node in state $C$ has degree $l$.
Then
\begin{equation}
    [ABC] = \sum_{j,k,l} [A_j B_k C_l]
\end{equation}
and similarly for $[A]$ and $[AB]$.
Let us also define $[S \bu] := [SA]$ where $A \in \{S,I\}$.
We derive a closure condition starting from the finest resolution, where we account for the degree of each node. 
The first approximation is:
\begin{equation}
    [A_j S_k I_l]
    \simeq (k - 1) [A_j S_k] \frac{[S_k I_l]}{k [S_k]}
    \simeq (k - 1) [A_j S_k] \frac{[S I_l]}{[S \bu]}
\end{equation}
where we assume that a degree $k$ susceptible node's neighbour is as likely to be infected as any other susceptible node's neighbour.
Intuitively, the first approximation is valid because we start with an $A_j S_k$ pair and each of the $k - 1$ additional edges connected to the $S_k$ node leads to an $I_l$ node with probability $[S_k I_l] / (k [S_k])$.
The second approximation follows from the configuration model.
These will be used repeatedly in what follows.
Summing over the $l$ index alone we get:
\begin{equation}
   [A_jS_kI] = \sum_{l} [A_j S_k I_l]
   \simeq \sum_{l} (k - 1) [A_j S_k] \frac{[S I_l]}{[S \bu]}
   \simeq (k - 1) [A_j S_k] \frac{[S I]}{[S \bu]}
\end{equation}
A similar approximation to $[A_j S_k]$ and summation over $j$ leads to $[AS_k]=[AS]\times (k[S_k]/[S\bu])$.
This, in turn, leads to:
\begin{equation}
    [A S_k I] = \sum_{j} [A_j S_k I]
    =\sum_{j} (k - 1) [A_j S_k] \frac{[S I]}{[S \bu]}
    \simeq (k - 1) [A S_k] \frac{[S I]}{[S \bu]}
    \simeq (k - 1) k[S_k] \frac{[A S][S I]}{[S \bu]^2}
\end{equation}
Finally, summing over $k$ leads to:
\begin{equation}
    \label{eq:cls0}
    [ASI]=\sum_{k}[AS_kI]= \sum_{k}(k-1)k[S_k]\frac{[AS][SI]}{[S\bu]^2}\simeq \frac{[AS][SI]}{[S\bu]^2}\sum_{k}(k-1)k[S_k]
\end{equation}
It remains to handle $\sum_k k (k - 1) [S_k]$.

Recall that the variable $\theta$  ($x_\theta$ in the DSA model) is the probability that infection has not crossed a randomly chosen edge, and it decreases over time.
A randomly chosen node of degree $k$ will remain susceptible with probability $\theta^k$, so $[S_k] \simeq N p_k \theta^k$ where $p_k$ is the probability mass on $k$ in the degree distribution.
In terms of the degree distribution PGF $\psi$,  for large $N$ we have approximately (see, for instance, \cite{jacobsen2018large}):
\begin{equation}
    \begin{aligned}
        [S] & \simeq N \psi(\theta), \\
        [S \bu] \simeq \sum_k k N p_k \theta^k 
        & = N \theta \psi'(\theta), \\ 
        \sum_{k} (k - 1) k [S_k]
        &\simeq N \theta^2 \psi''(\theta)
    \end{aligned}
\end{equation}
Using these and  \eqref{eq:cls0}  leads to:
\begin{equation}
    \label{eq:cls1}
    [ASI]\simeq\frac{\psi''(\theta)\psi(\theta)}{\psi'(\theta)^2}\frac{[AS][SI]}{[S]}
\end{equation}
Because $\theta$ is a new dynamic variable, an equation for it is needed.
However, we need no more equations as long as
\begin{equation}
    \frac{\psi''(\theta) \psi(\theta)}{\psi'(\theta)^2}
    = \kappa = \text{ constant},
    \label{eq:constant}
\end{equation}
in which case the dependency on $\theta$ is curtailed. 
If $\theta$ is the probability that disease has not crossed a randomly chosen edge, then it follows from equations~\eqref{eq:meandeg} and~\eqref{eq:meanxdeg} that the left-hand side of equation~\eqref{eq:constant} is the mean excess degree of susceptible nodes divided by their mean degree.
Therefore, the  condition above simply implies that this ratio remains constant as the susceptible nodes are depleted over time. 
Below, we show that the networks for which this property holds can be explicitly enumerated.
Remarkably, for such networks,  \eqref{eq:constant} is equivalent to the exact  closure---that is to the asymptotic ($N \to \infty $) equality in~\eqref{eq:cls1}.

\subsection{Poisson-type distributions}
Assuming that $\psi,\psi^\prime >0$ on the domain  $[0,1]$, the condition  \eqref{eq:constant} can be rewritten as 
\begin{equation}
    \frac{\psi''(u)}{\psi'(u)} =\ka \frac{\psi'(u)}{\psi(u)}.
    \label{eq:pgfeq}
\end{equation}
for any $u \in [0,1]$.
Upon integrating, we get the first-order differential equation    
\begin{equation}
    \psi^{\prime}(u) = \alpha\, \psi(u)^\ka
    \label{eq:1}
\end{equation} 
for arbitrary constants $\alpha > 0$ and $\ka > 0$.
Because $\psi$ is analytic, the equation above is defined beyond the original domain---in particular in the small right-side neighbourhood of the natural initial condition $\psi(1) = 1$.  
 
Table~\ref{tab:1} below presents a family consisting of three distributions whose PGFs satisfy the ODE~\eqref{eq:1}.
We refer to these distributions as {\em Poisson-type} (PT) distributions.
It turns out that being the PGF of a PT distribution is necessary and sufficient for \eqref{eq:1} to hold:

\begin{theorem}[Characterization of the Poisson-type distributions] The PGF of a random variable satisfies \eqref{eq:1} iff the random variable belongs to the PT family listed in Table~\ref{tab:1}.
\end{theorem}
\begin{proof}
\noindent If $\ka = 1$, the ODE given by \eqref{eq:1} and the PGF condition $\psi(1) = 1$ imply that $\psi(u) = e^{u (\alpha - 1)}$, which is the PGF of the Poisson random variable $\mathrm{POI}(\alpha)$ in Table~\ref{tab:1}.
If $\ka \ne 1$, separating variables and integrating gives us
\begin{equation}
    \frac{\psi(u)^{1 - \ka}}{1 - \ka} = \alpha u+c,
\end{equation}
for some  constant $c$.
Taking into account the condition $\psi(1) = 1$, we get
\begin{equation}
    \psi(u) = [\alpha (1 - \ka) (u - 1) + 1]^{\frac{1}{1 - \ka}}.
\end{equation}
Consider now separately the  cases when $\ka < 1$ and $\ka > 1$.
\begin{itemize}
    \item \textbf{Case $\ka \in (0, 1)$:}
    Because $\psi^{(s)}(0) \ge 0$ for each integer $s \ge 0$, we must have $n = (1 - \ka)^{-1}$ is a positive integer and $\alpha (1 - \ka) \leq 1$.
    Writing
    \begin{equation}
        \psi(u) =
        \big[1 - \alpha (1 - \ka) + \alpha (1 - \ka) u\big]^n,
    \end{equation}
    we recognise  $\psi(u)$ as the PGF of the binomial random variable $\mathrm{BINOM}(n,p)$ with $p=\alpha(1 - \ka)$. 
    
    The special case $p = 1$, which corresponds to an $n$-regular  degree distribution, violates our initial assumptions that $\psi(0) > 0$ and $\psi'(0) > 0$, but we can check directly that the $\mathrm{BINOM}(n, 1)$ PGF $\psi(u) = u^n$ satisfies \eqref{eq:1}.
    When $u \neq 0$, it also satisfies~\eqref{eq:constant} and~\eqref{eq:pgfeq}.
    Hence, we allow for $\mathrm{BINOM}(n, p)$ with $p = 1$ in Table~\ref{tab:1}.
    
    \item \textbf{Case $\ka > 1$:} Writing 
    \begin{equation}
        \psi(u) =
        \left[\frac{\frac{1}{\alpha (\ka - 1) + 1}}{1 - \frac{\alpha (k - 1)}{\alpha (\ka - 1) + 1}u}\right]^{\frac{1}{\ka-1}},
    \end{equation}
    we recognise $\psi(u)$ as the PGF of the  negative binomial distribution $\mathrm{NB}(r,p)$ with $r=\frac{1}{\ka - 1}$  and $p=\frac{\alpha(\ka-1)}{\alpha(\ka - 1) + 1}$.
    Note that here necessarily $p<1$.
\end{itemize}
From considering all possible values of $\kappa$, it follows that there are only three PGF solutions to equation \eqref{eq:1}, which correspond to the families of random variables listed in Table~\ref{tab:1}.
\end{proof}

\begin{table}
\centering
\caption{The PT random variables whose PGFs satisfy \eqref{eq:1}. }
\label{tab:1}
    \begin{tabular}{llc}
        \toprule
        Condition           & Family 
        & Parameters \\
        \midrule
        $\ka \in (0, 1)$
        & Binomial: $\mathrm{BINOM}(n,p)$ 
        &  $n=\frac{1}{1-\ka}$,  $p=\alpha(1-\ka)$ \\[5pt]
        $\ka = 1$
        & Poisson: $\mathrm{POI}(\lambda)$ 
        & $\lambda=\alpha$ \\[5pt]
        $\ka>1$
        & Negative binomial: $\mathrm{NB}(r,p)$ 
        &  $r=\frac{1}{\ka-1}$,  $p=\frac{\alpha(\ka-1)}{\alpha(\ka-1)+1}$\\ 
        \bottomrule
    \end{tabular}
\end{table}
\subsection{Closure theorem and models equivalence}
We are now in a position to state the main result on the exactness of the pairwise closure.
Here, we use ``exactness'' in the sense defined in~\cite{jacobsen2018large}, or, equivalently, in~\cite{janson2014law}.
In both cases, the notion implies that the appropriately scaled stochastic vector of susceptibles, infecteds, and recovereds tends in an appropriate sense to a deterministic vector whose components are described by the system of ordinary differential equations given by~\eqref{eq:VOrig} or~\eqref{eq:GR}.
Yet another equivalent definition of the exact closure is that the equality in the triple approximation condition~\eqref{eq:cls2} holds upon dividing both its sides by $N$ and taking the limit $N \to \infty$.

\begin{theorem}[Exact pairwise closure] The closure condition 
\begin{equation}\label{eq:cls2}
    [ASI]\simeq\kappa \frac{[AS][SI]}{[S]}    
\end{equation}
for $A \in \{S, I\}$ in the pairwise model given by system~\eqref{eq:PW_unclosed} is exact (that is,  equality in \eqref{eq:cls2} with both sides  multiplied by $N^{-1}$ holds asymptotically as $N\to \infty$) iff the underlying configuration model network has a Poisson-type (PT) degree distribution.
Furthermore, $\kappa = (n - 1) / n < 1$, $\kappa = 1$, or $\kappa = (r + 1) / r > 1$ if the degree distribution is $\mathrm{BINOM}(n,p)$, $\mathrm{POI}(\lambda)$, or $\mathrm{NB}(r,p)$, respectively.

\end{theorem}
\begin{proof}
Consider first evaluation of the $\psi''\left(x_{\theta}\right)/ \psi'\left(x_{\theta}\right)^{2}$ term.
In Table~\ref{tab:PT-type-equiv}, we show that this term is equivalent to $\kappa/x_{S}$ for all PT distributions.
With this in mind, we are ready to show the equivalence between the limiting pairwise model and the DSA models under \eqref{eq:cls2}.

\begin{table}[h!]
    \centering
    \caption{Resolving the $ \psi''\left(x_{\theta}\right)/ \psi'\left(x_{\theta}\right)^{2}$ term in the DSA equations for binomial, Poisson, and negative binomial distributions.
    Note that $\psi(x_{\theta})=x_S$.}
    \label{tab:PT-type-equiv}
    \begin{tabular}{cccc}
        \toprule
                        & \text {Binomial}
                        & \text {Poisson}       
                        & \text {Negative binomial}
                        \\[5pt]
        \midrule
        \text {Parameter(s)}    
                        & $(n, p)$ 
                        & $\lambda$             
                        & $(r, p)$ 
                        \\[5pt]
        $\psi(x)$               
                        & $(1 - p + p x)^n$
                        & $e^{\lambda(x - 1)}$  
                        & $\big(\frac{1-p}{1-p x}\big)^{r}$ 
                        \\[5pt]
        $\psi'(x)$              
                        & $n p (1 - p + p x)^{n - 1}$ 
                        & $\lambda e^{\lambda(x - 1)}$ 
                        & $\frac{r p (1 - p)^r}{(1 - p x)^{r + 1}}$ 
                        \\[5pt]
        $\psi''(x)$             
                        & $n (n - 1) p^2 (1 - p + p x)^{n-2}$ 
                        & $\lambda^2 e^{\lambda (x - 1)}$ 
                        & $\frac{r (r + 1) p^2 (1 - p)^r}
                            {(1 - p x)^{r + 2}}$ 
                        \\[5pt]
        $\frac{\psi''(x_{\theta})}{\psi'(x_{\theta})^2}$ 
                        & $\frac{n-1}{n} \times \frac{1}{x_S}$
                        & $1 \times \frac{1}{x_S}$ 
                        & $\frac{r + 1}{r} \times \frac{1}{x_S}$
                        \\[5pt]
        $\kappa$                
                        & $\frac{n - 1}{n}$ 
                        & $1$ 
                        & $\frac{r + 1}{r}$
                        \\
        \bottomrule
    \end{tabular}
\end{table}


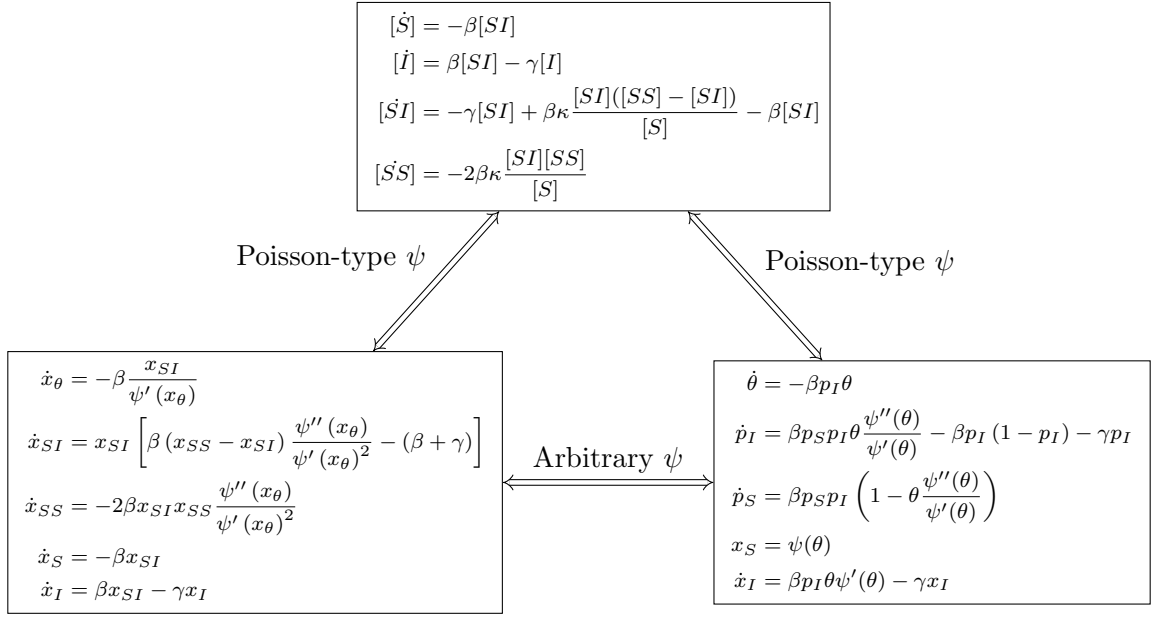
\begin{figure}[h!]
\centering
\begin{tikzpicture}
\node [rectangle,draw] (a) at (0, 5) {
\scriptsize{
$\begin{aligned}
[\dot{S}] &= -\beta[SI]\\
[\dot{I}] &= \beta[SI]-\gamma[I]\\
[\dot{SI}] &= -\gamma[SI] + \beta\ka\frac{[SI]([SS]-[SI])}{[S]} - \beta[SI]\\
[\dot{SS}] &= -2\beta\ka\frac{[SI][SS]}{[S]}
\end{aligned}$}}; 
\node [rectangle,draw] (b) at (-4.5, 0){
\scriptsize{
$\begin{aligned}
\dot{x}_{\theta}&= -\beta \frac{x_{SI}}{ \psi'\left(x_{\theta}\right)}\\
\dot{x}_{SI}&= x_{S I}\left[\beta\left(x_{S S}-x_{S I}\right) \frac{ \psi''\left(x_{\theta}\right)}{ \psi'\left(x_{\theta}\right)^{2}}-(\beta+\gamma)\right] \\
\dot{x}_{SS}&=  -2 \beta x_{S I} x_{S S} \frac{ \psi''\left(x_{\theta}\right)}{ \psi'\left(x_{\theta}\right)^{2}}\\
\dot{x}_{S}&= -\beta x_{S I}\\
\dot{x}_{I}&=  \beta x_{S I}-\gamma x_{I} 
\end{aligned}$
}
};
\node [rectangle,draw] (c) at (4.5, 0){
\scriptsize{
$\begin{aligned}
\dot{\theta}&=-\beta p_{I} \theta\\
\dot{p}_{I}&=\beta p_{S} p_{I} \theta \frac{\psi^{\prime \prime}(\theta)}{\psi^{\prime}(\theta)}-\beta p_{I}\left(1-p_{I}\right)-\gamma p_{I}\\
\dot{p}_{S}&=\beta p_{S} p_{I}\left(1-\theta \frac{\psi^{\prime \prime}(\theta)}{\psi^{\prime}(\theta)}\right)\\
x_S&=\psi(\theta)\\
\dot{x}_I&=\beta p_{I} \theta \psi^{\prime}(\theta)-\gamma x_I
\end{aligned}$
}
};
    \draw[implies-implies,double equal sign distance] (a) -- (b) node[midway, above left] {Poisson-type $\psi$};
    \draw[implies-implies,double equal sign distance] (a) -- (c) node[midway, above right] {Poisson-type $\psi$};
 \draw[implies-implies,double equal sign distance] (b) -- (c) node[midway,above] {Arbitrary $\psi$};
\end{tikzpicture}
\caption{Summary of model equivalence results. 
The top is the pairwise model, the bottom left is the DSA model, and the bottom right is the Volz model.}
\label{Fig:ModelEquiv}
\end{figure}

Let us show equivalence between the evolution equations for $[SI]$ and $x_{SI}$. 
Recall that  $x_{A} =\lim_{N\to \infty} [A]/N$ and $x_{AB} =\lim_{N\to \infty} [AB]/N$ and that these limits exist uniformly in probability over any finite time interval~\cite{jacobsen2018large}. 
From the equation for ${SI}$ in the pairwise model and~\eqref{eq:cls2} 
\begin{equation*}
    \begin{aligned}
        \dot{[SI]}
        & = \beta \big([SSI] - [ISI]\big) - [SI](\beta + \gamma) \\
        & = [SI] \left[\beta \big([SS] - [SI]\big) \frac{\kappa}{[S]} - (\beta + \gamma)\right].
    \end{aligned}
    \label{eq:ss1}
\end{equation*}
Dividing both sides of the last equation by $N$, taking the limit $N \to \infty$ (which can be done in view of the appropriate LLN, see~\cite{jacobsen2018large}),  and using the fact  that  $\psi''\left(x_{\theta}\right) /  \psi'\left(x_{\theta}\right)^2 = \kappa / x_S$,  we arrive at: 
\begin{equation*}
    \begin{aligned}
         \dot{x}_{SI} & =  x_{SI} \bigg[\beta\left(x_{S S}-x_{S I}\right) \frac{\kappa}{x_{S}}-(\beta+\gamma)\bigg] \\
        &=  x_{SI}\left[\beta\left( x_{S S} -  x_{S I}\right)   \frac{\psi''\left(x_{\theta}\right)}{\psi'\left(x_{\theta}\right)^2} - (\beta + \gamma)\right] 
    \end{aligned}
\end{equation*}
which is identical to the equation for $x_{SI}$ in the DSA model.
We note that, when the degree distribution is PT, the equation for $x_{\theta}$ is no longer needed and the equivalences of the remaining equations follow similarly as above.
For instance, from \eqref{eq:cls2} and $\dot{[SS]}  = -2 \beta [SSI]$ it follows that $\dot{x}_{SS} = -2 \beta \kappa x_{SS}  x_{SI} /x_S = -2 \beta  x_{SS}  x_{SI}\psi''\left(x_{\theta}\right)/\psi'\left(x_{\theta}\right)^2$. 
The exactness of the DSA model~\cite{jacobsen2018large} as the limit of the stochastic SIR model on configuration network, implies thus the exactness of the scaled pairwise model (and \eqref{eq:cls2}) as  $N\to \infty$.
\end{proof}
  
Figure~\ref{Fig:ModelEquiv} summarises the model equivalences.
The equivalence of the Volz model and the DSA model for arbitrary degree distributions is shown in the Appendix.

The top row of Figure~\ref{fig:Closure_for_PT_type_nets} shows numerical evidence of the exactness of the closure in the pairwise model for PT-type networks.
For PT-type networks, the agreement between the pairwise model and the expected value of explicit stochastic simulations is excellent.
The DSA model continues to work well for non PT-type networks (see bottom row of Figure~\ref{fig:Closure_for_PT_type_nets}), and it is clear that $\kappa$ is not constant in time in this case.
As expected, this means that none of the three possible closures work.
In the left panel of the bottom row we plot the output from the pairwise model for $\kappa = (n - 1) / n$ (dashed line) and $\kappa = 1$ (dotted line).
Both underestimate prevalence which in this case is driven by the 20\% of highly connected nodes.
This is captured poorly by both closures.

\begin{figure}[]
     \centering
     \includegraphics[width=0.30\textwidth]{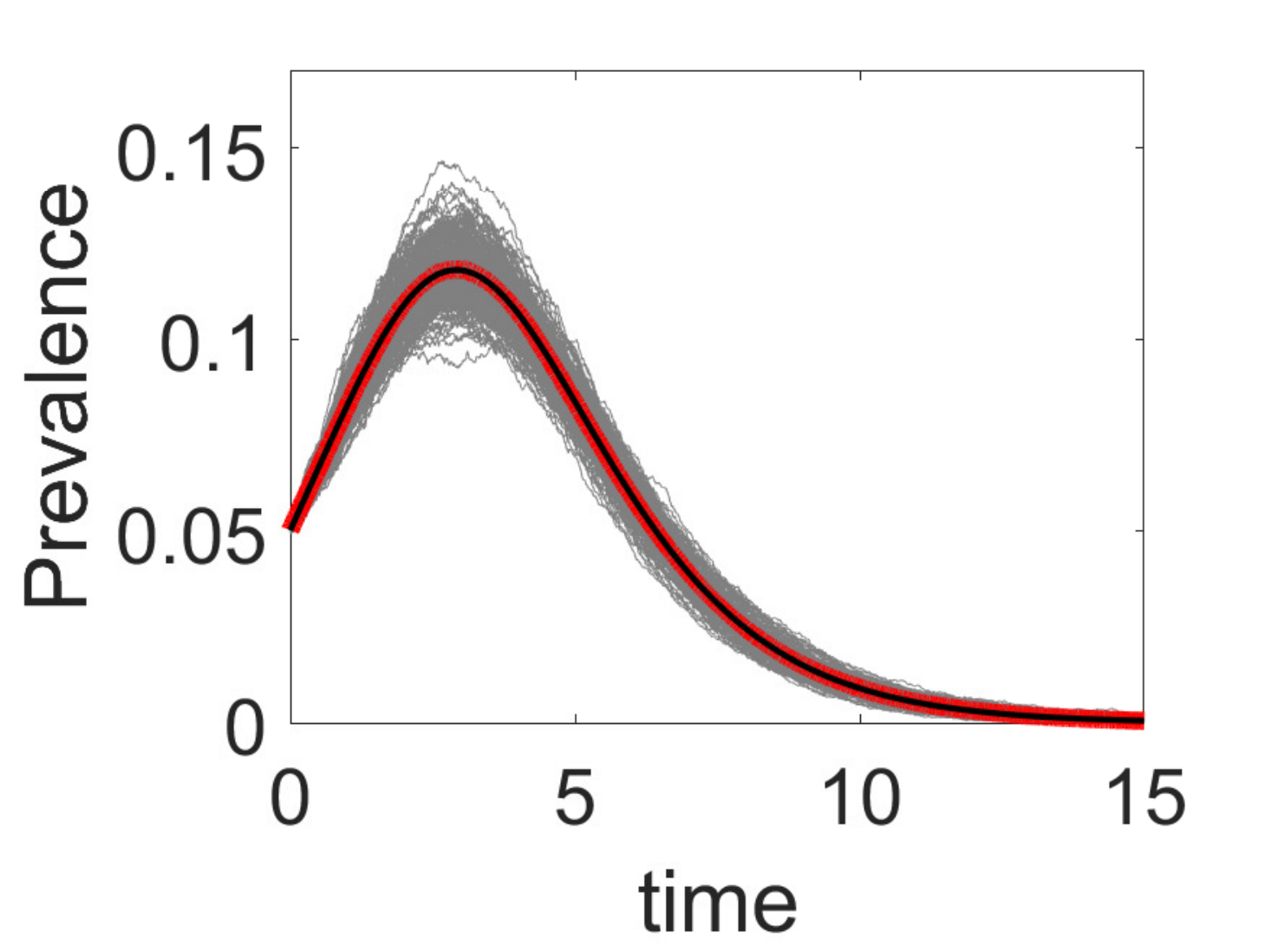}
     \includegraphics[width=0.30\textwidth]{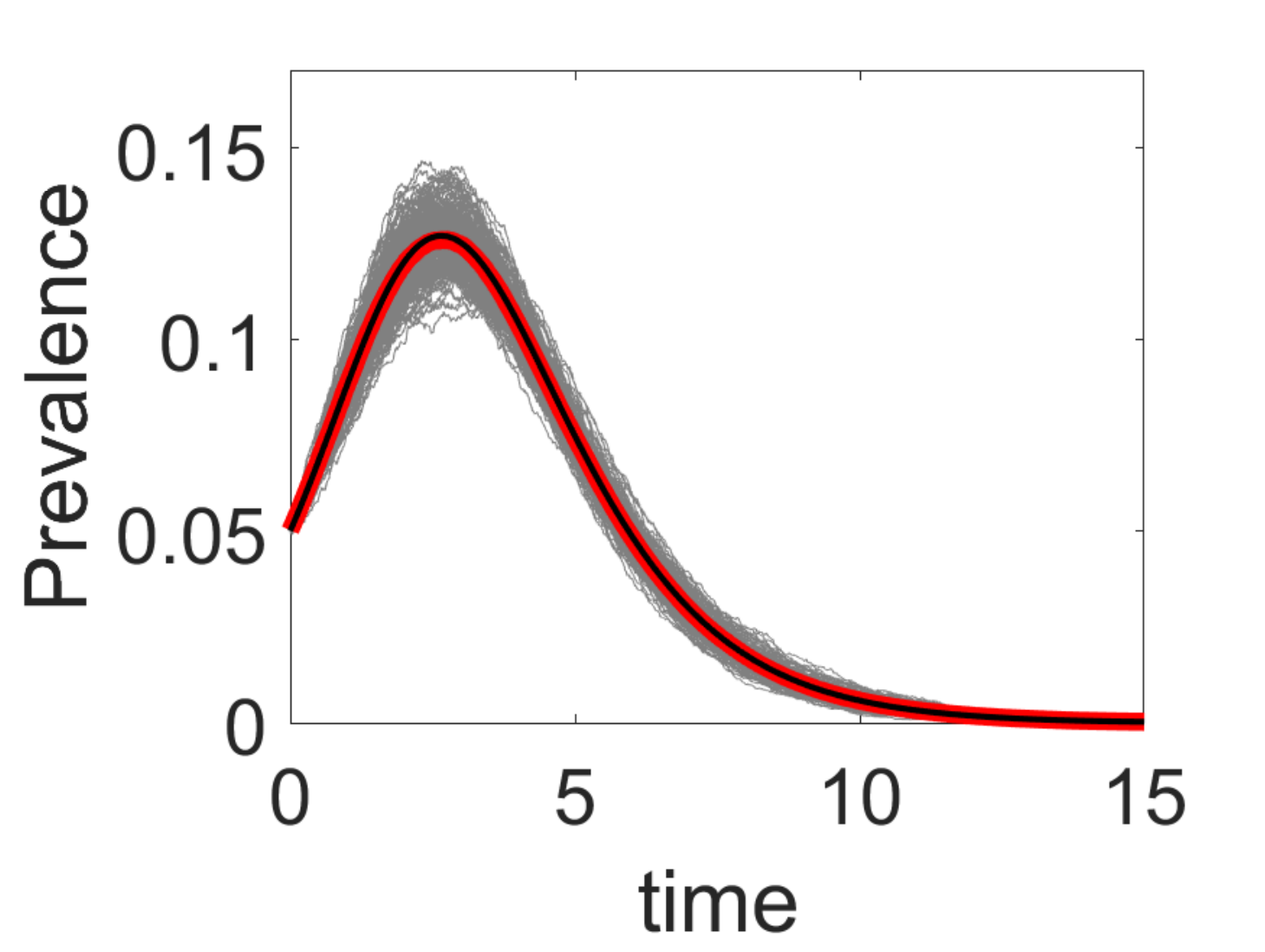}
     \includegraphics[width=0.30\textwidth]{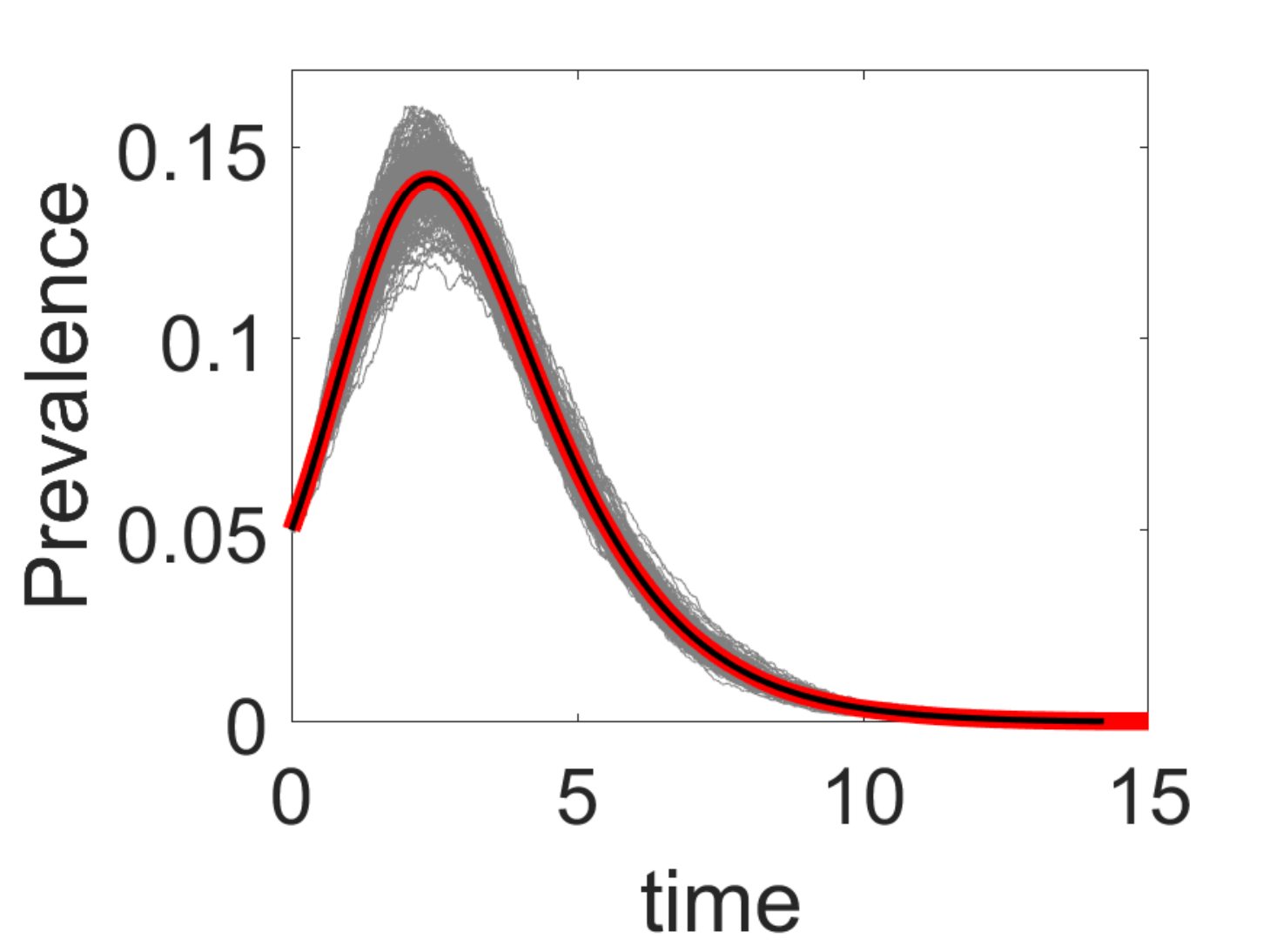}
    \includegraphics[width=0.40\textwidth]{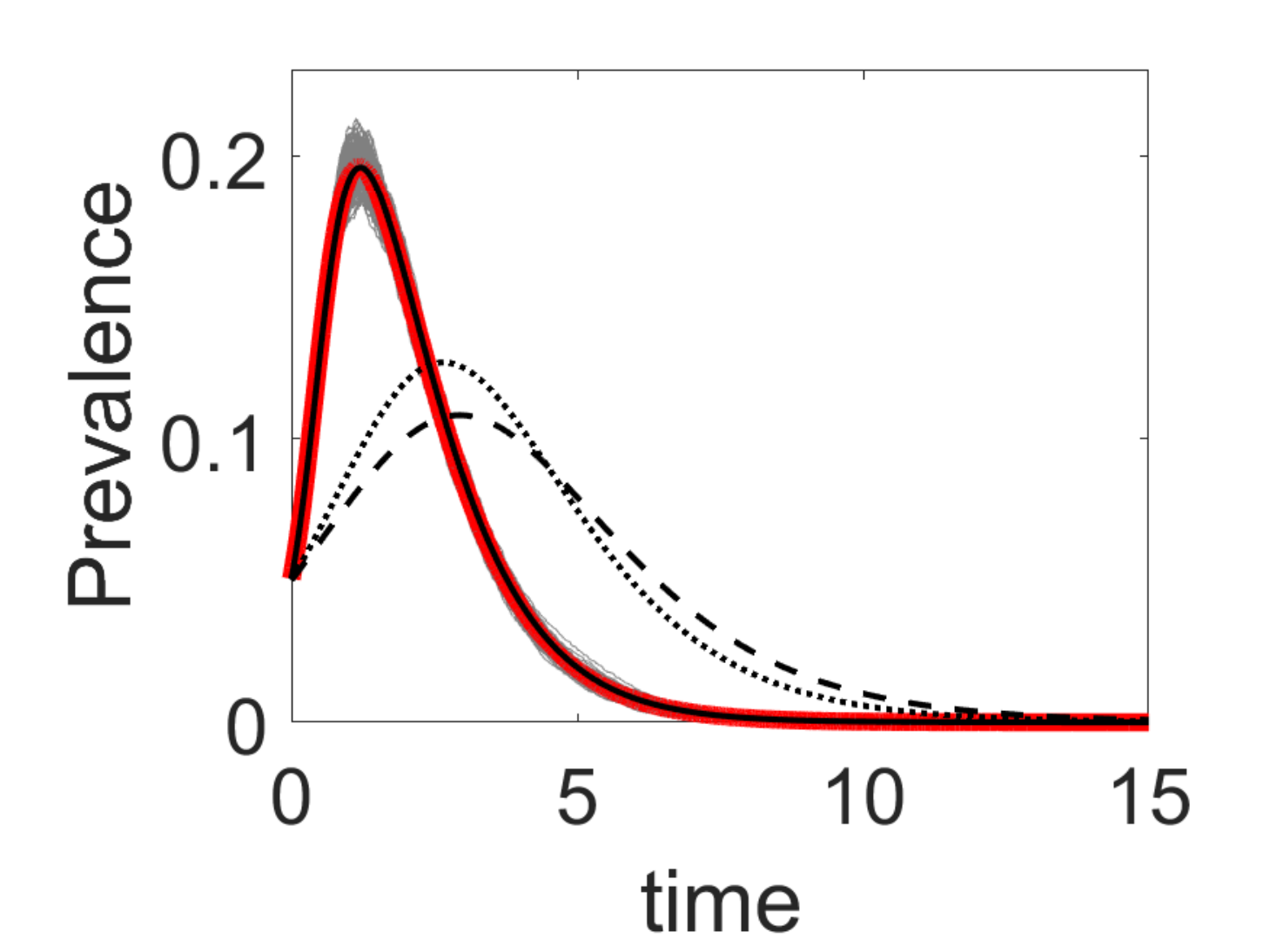}
    \includegraphics[width=0.40\textwidth]{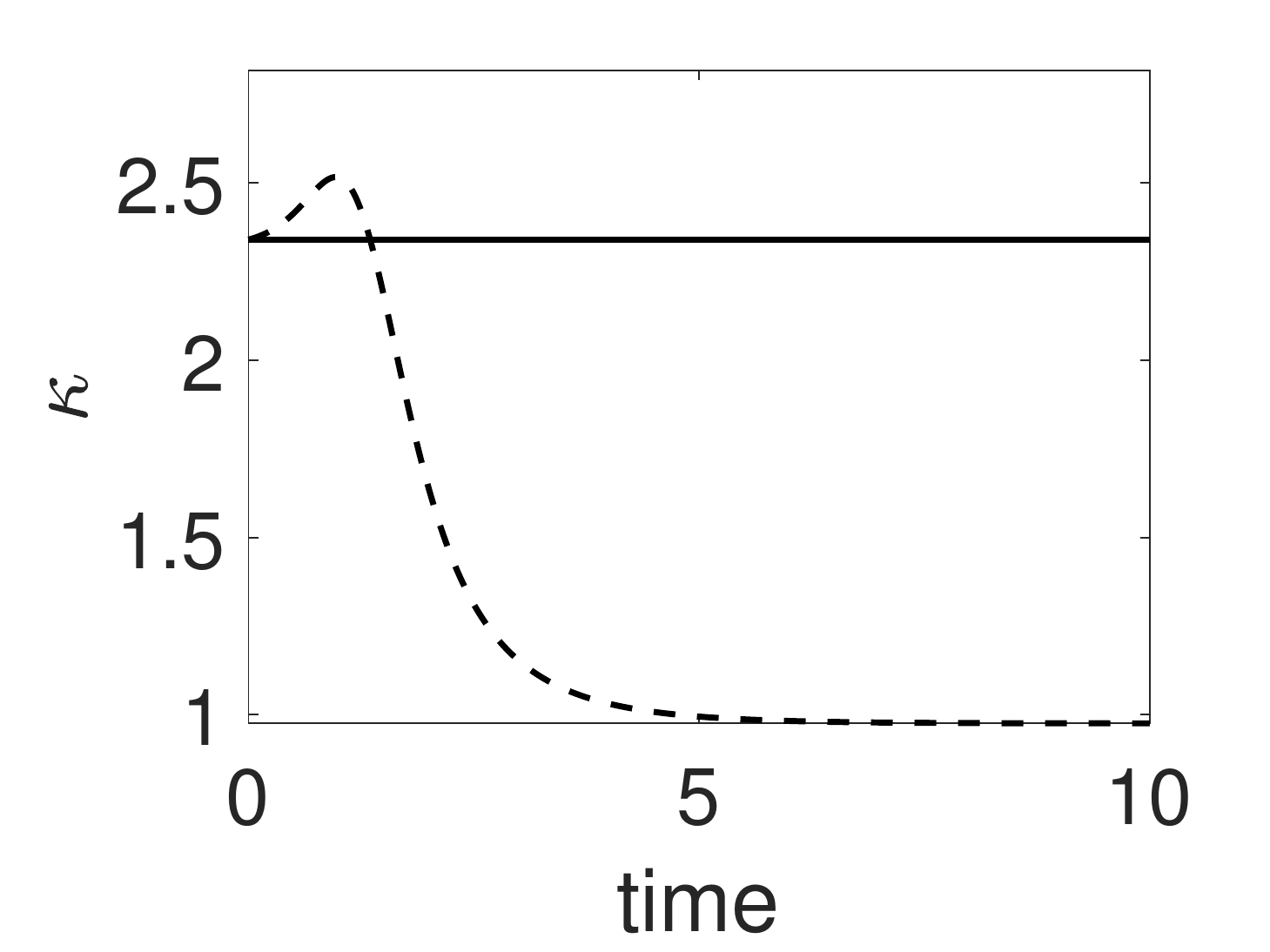}
     \caption{
        {\bf Top row:} 
        From left to right, epidemics on networks with regular (each node has $n = 6$ links), Poisson ($\lambda = 10$) and negative binomial ($r = 10$ and $p = 1/2$) degree distributions are plotted, respectively.
        Individual stochastic realisations are plotted with thin grey lines, their mean with the thick red line, and the solution of the corresponding pairwise model with a solid black line. 
        {\bf Bottom left:}
        Epidemics on a network where 80\% of nodes have degree 4 and the rest have degree 34.
        The DSA model (solid black line) is used to match the average epidemics.
        The pairwise model closures with $\kappa = (n - 1) / n$ (dashed line) and $\kappa = 1$ (dotted line) are also plotted.
        {\bf Bottom right:}
        Plot showing that, for a non-PT-type network (like the network with two distinct value for the degree of nodes), $\kappa$ is not constant in time.
        The value of $\kappa$ at time $t=0$ (solid constant line) is short lived as shown $\kappa(t)$ (dashed line) as given in equation~\eqref{eq:constant}.
        Other parameters are: $N = 10000$ nodes, recovery rate $\gamma = 1$, and per-contact transmission rate $\beta = 0.4$ for the regular network and $\beta = 0.2$ for the networks with Poisson, negative binomial, and mixed degree distributions.
        Epidemics start with 250 infected nodes chosen at random and only epidemics reaching 500 infected individuals are retained.
        We average over 15 network realisations and 15 epidemics on each network.
    }
    \label{fig:Closure_for_PT_type_nets}
\end{figure}


\section{Survival analysis perspective}
\label{sec:DSAper}


The exact closure condition implies that, under the assumption of a PT-type degree distribution, the pairwise and DSA models are equivalent.
One of the benefits of this equivalence is that the pairwise model shares the statistical interpretation of the DSA model.
Indeed, as shown in~\cite{khudabukhsh2020survival} (see also discussion with examples in~\cite{choi2019modeling, bastian2020throwing,di2022dynamic,vossler2022analysis}), we can interpret the system of equations~\eqref{eq:GR} in terms of a statistical model for times to infection.
To this end, as in~\cite{khudabukhsh2020survival}, we may consider $S_t := x_S(t)$ as the survival probability of a typical node (i.e., the probability  that a typical node who was susceptible at time $t = 0$ remains susceptible at time $t > 0$).
Note that  $S_0 = 1$ follows from the assumed initial conditions~\eqref{eq:GR_ic} and that $S_\infty > 0$, so $S_t$ is an improper survival function.
In this section, we show how to derive a single autonomous differential equation for $S_t$ (or $x_S$) that allows numerical calculation of the survival probability for any $t\in [0, \infty)$ solely in terms of the network model parameters.    
We achieve this in several steps:
First, we derive an integral that relates $x_{SS}$ and $x_{S}$.
Because $\psi''\left(x_{\theta}\right) /  \psi'\left(x_{\theta}\right)^2 = \kappa / x_S$ under the pairwise closure condition, we obtain:
\begin{equation*}
    \frac{\dot{x}_{SS}}{\dot{x}_S}= 2 \ka \frac{x_{SS}}{x_S}
\end{equation*}
Integrating this and using the initial conditions $x_S(0) = 1$ and $x_{SS}(0) = \mu$ leads to:
\begin{equation*}
    x_{SS}(t) = \mu x_{S}(t)^{2 \ka}
    \label{eq:SSFirstIntegral}
\end{equation*}
Second, the equations for $x_{S}$ can be rewritten as:
\begin{equation*}
    \dot{x}_{S} = - \beta \frac{x_{SI}}{x_{S}}x_{S}
    = -\beta x_D x_S
\end{equation*}
where $x_{D}=x_{SI}/x_{S}$ is considered a new variable for which an evolution equation is needed.
Considering the derivative of $x_{D}$, and plugging in the expressions for $\dot{x}_{SI}$ and $\dot{x}_{S}$, we obtain:
\begin{equation*}
    \begin{aligned}
        \dot{x}_{D}
        &=\frac{\dot{x}_{SI}x_{s}-x_{SI}\dot{x}_{S}}{x_{S}^2} \\
        &= \frac{\beta \kappa x_{SS}x_{SI}-\beta \kappa x_{SI} x_{SI} - (\beta + \gamma) x_{SI} x_{S} + \beta x_{SI}^2}{x_{S}^2} \\
        &= \beta \kappa \mu x_{s}^{2\kappa-1} \left(\frac{x_{SI}}{x_{S}}\right) - \beta \kappa \left(\frac{x_{SI}}{x_{S}}\right)^2 - (\beta + \gamma) \left(\frac{x_{SI}}{x_{S}}\right) + \beta \left(\frac{x_{SI}}{x_{S}}\right)^2 \\
        &= \beta (1 - \kappa) x_{D}^2 + \left(\beta \kappa \mu x_{S}^{2\kappa - 1} - (\beta + \gamma)\right) x_{D}
    \end{aligned}
\end{equation*}
Given that the equations for $x_{SI}$ and $x_{SS}$ no longer depend on $x_{\theta}$, the system now simplifies to three key equations:
\begin{equation*}
    \begin{aligned}
        \dot{x}_{S} &= -\beta x_{D}x_{S}, \\
        \dot{x}_{I} &= \beta x_{D} x_{S} - \gamma x_{I}, \\
        \dot{x}_{D} &= \beta (1 - \kappa) x_{D}^2 + \left(\beta \kappa \mu x_{S}^{2\ka - 1} - (\beta + \gamma)\right) x_{D} \\
    \end{aligned}
\end{equation*}
Finally, we can manipulate these equations further.
In particular, looking at 
\begin{equation*}
    \frac{\dot{x}_{D}}{\dot{x}_{S}} + (1 - \ka) \frac{x_{D}}{x_{S}} 
    = -\kappa \mu x_{S}^{2\kappa-2} + \frac{\beta + \gamma}{\beta} \frac{1}{x_{S}},
\end{equation*}
and considering $x_{D}$ as a function of $x_{S}$, we can use an integrating factor.
This leads to:
\begin{equation} \label{eq:StKeqOne}
    -\dot{S}_t = 
    \begin{cases}
        \tilde{\beta}(1-S_{t}^{\kappa})S_{t}^{\kappa}+\frac{\tilde{\gamma}}{1-\kappa}S_{t}(1-S_{t}^{\kappa-1})+\tilde{\rho} S_{t}^{\kappa} 
        &\text{if } \kappa\ne 1, \\
        \tilde{\beta} (S_t - S_{t}^{2}) + \tilde{\gamma} S_{t} \log S_{t} + \tilde{\rho} S_{t}
        &\text{if } \kappa = 1 ,
    \end{cases}
   \end{equation}
where  we now replaced $x_{S}(t)$ by $S_{t}$ and let $\tilde{\rho} = \beta \mu \rho$, $\tilde{\gamma}=\beta+\gamma$, and $\tilde{\beta}=\mu \beta$.
As already noted, the initial condition 
inherited from~\eqref{eq:GR_ic} is $S_0 = 1$.
Because necessarily $\dot{S}_\infty=0$, equation~\eqref{eq:StKeqOne}  implies that the limiting value $S_\infty>0$ has to satisfy
\begin{equation*}\label{eq:FinalSize}
   \begin{aligned}
        \tilde{\beta}(1-S_{\infty}^{\kappa})+\tilde{\rho} = \frac{\tilde{\gamma}}{1-\kappa}(1-S_{\infty}^{1-\kappa})
        &\qquad \text{if } \kappa\ne 1, \\
        \tilde{\beta} (1 - S_{\infty}) + \tilde{\rho} = -\, \tilde{\gamma}\,  \log S_{\infty} 
        &\qquad \text{if } \kappa = 1.
    \end{aligned}
\end{equation*}
It is of interest to note that when the degree distribution is Poisson ($\kappa = 1$), then equation~\eqref{eq:StKeqOne} is identical to that known from mass-action SIR dynamics.

This analysis shows that  the dynamics of an SIR epidemic 
on a configuration model network with a PT degree distribution can be 
summarised with  a single self-contained survival equation describing the evolution of survival probability $S_t$.
This leads to the following interesting statistical consideration that was already noted for mass-action SIR models in~\cite{khudabukhsh2020survival,di2022dynamic,vossler2022analysis}.
Assuming that, over a time interval $[0, T]$ where $T \le \infty$, we observe the times of infection $(t_1, \ldots, t_k)$ of a randomly selected set of $k$ initially susceptible nodes of our network, we may write the approximate log-likelihood function as 
\begin{equation*}
    \ell(\tilde{\beta},\tilde{\gamma},\tilde{\rho}\vert t_1,\ldots,t_k) =\sum_{i=1}^k \log S_{t_i}-k\log(1-S_T)
\end{equation*}
To obtain quantities other than $S_t$, evaluation of additional  ODEs is needed as discussed, for instance, in~\cite{khudabukhsh2020survival} or~\cite{hosp22}.
Let us also note that, since the DSA and Volz models are equivalent (see Lemma~\ref{lem:equiv} in the Appendix), the representation $S_t$ in ~\eqref{eq:StKeqOne} can be similarly  derived directly from  Volz's model~\eqref{eq:VOrig}.

\section{Discussion}


Over the last two decades, two types of disease network models have emerged as particularly relevant in many practical applications (including the COVID-19 pandemic): the so-called pairwise~\cite{rand1999,Keeling1999} and edge-based~\cite{volz2008sir,miller2012edge} approaches.
More recently, a  version of an edge-based approach,  dubbed DSA, was proposed in~\cite{jacobsen2018large,khudabukhsh2020survival} to facilitate statistical inference.

In this paper, we have shown that the three approaches are equivalent and asymptotically exact under the assumption that the contact network underlying the spreada of disease is a configuration model random graph with one of the three Poisson-type (PT) degree distributions: Poisson, binomial, or negative binomial.
Perhaps more interestingly,  we have shown that the pairwise closure for an epidemic on a configuration model network is exact if and only if the ratio of mean excess degree to mean degree for susceptible nodes remains constant over time (as the susceptible nodes are depleted).
This condition holds if and only if the degree distribution is PT.
As an interesting corollary of our results, we obtained a single equation representation of the pairwise model that allows parameter estimation from time series data marginalised over the network degree distribution.
This finding is practically useful as it allows, for instance, statistical inference based solely on the disease incidence data as in the  classical, homogeneous SIR models.
Because these statistical methods are based on survival times in susceptible individuals, statistical inference can be based on observation of a random sample of the population.

\section*{Acknowledgements}
EK and GAR were partially supported by the US National Science Foundation (NSF) grant DMS-1853587.
The contents are solely the responsibility of the authors and do not represent the official views or policy of the NSF.
\clearpage

\appendix

\section{Summary of notation}
The following notation is used throughout the paper and in particular in the next section.
\begin{itemize}
\item $\beta:=$ the force of infection per infectious neighbour (i.e., the constant rate at which infectious nodes infect a neighbour).
 \item $\gamma:=$ the recovery rate (i.e., the constant rate at which infected nodes become recovered).
 \item $p_{k}:=$ the probability that a node will have degree $k$.
 \item $\psi(x):=$ the probability generating function for the degree distribution $\left\{p_{k}\right\}$.
 \item $x_S:=$ the proportion of nodes susceptible at time $t$.
 \item $x_I:=$ the proportion of nodes infectious at time $t$
 \item $x_{AB}:=$ the proportion of $AB$-type  edges at time $t$
 \item $\mathcal{A}_{X} :=$ the set of arcs ($i, j$) such that node $i$ is in set $X$.
 \item $M_{X} :=$ the proportion of arcs in set $\mathcal{A}_{X}$.
 \item $\mathcal{A}_{X Y} :=$ the set of arcs ($i, j$) such that $i \in X$ and $j \in Y$.
\item $M_{X Y} :=$ the proportion of arcs in set $\mathcal{A}_{XY}$.
\item $\theta :=$ the probability that infection has not crossed a randomly chosen edge at time $t$, also denoted by $x_{\theta}$ in some of the models.
\item $p_{I} := M_{SI} /M_{S}$,  the probability that an arc $(i, j)$ with a susceptible $i$ has an infectious $j$.
\item $p_{S} := M_{SS} /M_{S}$,  the probability that an arc $(i, j)$ with a susceptible $i$ has a susceptible $j$.
\end{itemize}

\section{Equivalence of   Volz's and DSA models}

\label{app:volzEqDSA}

\begin{lemma}\label{lem:equiv}
The Volz model~\eqref{eq:VOrig}--\eqref{eq:Vol_ic} and the 
DSA model~\eqref{eq:GR}--\eqref{eq:GR_ic} are equivalent.
\end{lemma}
\begin{proof}
While our starting point is the system of Volz's original set of equations~\eqref{eq:VOrig}, it is useful to recast these over a different state space in order to show equivalence with the DSA system in equations~\eqref{eq:GR}.
Here, we move from the state space in terms of ($\theta, p_{I}, p_{S}$) to a state space in terms of ($\theta=x_{\theta}, x_{SI}, x_{SS}$), where $x_{SI}$ and $x_{SS}$ are simply the limiting counts of all $SI$ and $SS$ -type edges (counted in both directions) scaled by $N$ (the number of nodes in the network) with the limit taken as $N \to \infty$ . 

As alluded to above, $\theta$ and $x_{\theta}$ have exactly the same interpretation, but we denote them differently to differentiate consistently between models over different state spaces.
We now proceed to show how one moves from the original Volz model to the DSA equations. We start by showing that the first equation in \eqref{eq:VOrig} is equivalent to the first equation in \eqref{eq:GR}. Starting from \eqref{eq:VOrig}, and taking into account that $M_{SI}=x_{SI} / \psi'(1)$ and that $M_{S} = x_\theta \psi'(x_\theta) / \psi'(1)$ as shown in~\cite{volz2008sir}, we obtain
\begin{equation}
    \dot{\theta} 
    = -\beta p_{I}\theta 
    = -\beta \frac{M_{SI}}{M_S} \theta
    = -\beta \frac{\frac{x_{SI}}{\psi'(1)}}{\frac{x_{\theta} \psi'(x_{\theta})}{\psi'(1)}} x_{\theta}
    = -\beta \frac{x_{SI}}{\psi'(x_{\theta})}
    = \dot{x}_\theta.
\end{equation}

Showing the equivalence between the other equations requires an extra step. 
That is, the evolution equations for $p_I$ and $p_S$ need to be rewritten explicitly in terms of $SI$ and $SS$ edges.
In~\cite{volz2008sir}, it was shown that the equations for $M_{SI}$ can be written as
\begin{equation}
    \dot{M}_{SI}
    = \beta p_{I} \left(p_{S} - p_{I}\right) \theta^{2} \frac{\psi''(\theta)}{\psi'(1)} - (\beta + \gamma) M_{SI}.
    \label{eq:VOrigMSI}
\end{equation}
Using the expressions for $M_{SI}$ and $M_S$ in terms of the DSA model parameters above and the fact that $M_{SS} = x_ss / \psi'(1)$, we get $p_I = x_{SI} / (x_\theta \psi'(\theta))$ and $p_S = x_{SS} / (x_\theta \psi'(x_\theta))$.
Substituting these into~\eqref{eq:VOrigMSI}, we get
\begin{equation*}
    \begin{aligned}
        \dot{M}_{SI}
        &= \beta p_{I} \left(p_{S} - p_{I}\right) \theta^{2} \frac{\psi''(\theta)}{ \psi'(1)} - (\beta + \gamma) M_{SI} \\
        &= \beta \frac{x_{SI}}{x_\theta \psi'(x_\theta)} \left(\frac{x_{SS}}{x_\theta \psi'(x_\theta)} - \frac{x_{SI}}{x_\theta \psi'(x_\theta)}\right) \frac{x_\theta^2 \psi''(x_\theta)}{\psi'(1)} - (\beta + \gamma) \frac{x_{SI}}{\psi'(1)} \\
        &= \frac{x_{SI}}{\psi'(1)} \left[\beta (x_{SS} - x_{SI}) \frac{\psi''(x_\theta)}{\psi'(x_\theta)^2} - (\beta + \gamma)\right] \\
        &= \frac{\dot{x}_{SI}}{\psi'(1)}
    \end{aligned}
\end{equation*}
Thus, equation~\eqref{eq:VOrig} is equivalent to the equation for $x_{SI}$ in \eqref{eq:GR}.
Following~\cite{volz2008sir}, the evolution equation for $SS$ edges can be rreduced to 
\begin{equation}
    \dot{M}_{SS}
    = -2 \beta p_I p_S \theta^2 \frac{\psi''(\theta)}{\psi'(1)}     
    \label{eq:VOrigMSS}
\end{equation}
Making the same substitutions as before, we get
\begin{equation*}
    \begin{aligned}
        \dot{M}_{SS}
        &= -2 \beta \frac{x_{SI}}{x_\theta \psi'(x_\theta)} \frac{x_{SS}}{x_\theta \psi'(x_\theta)} \frac{x_\theta^2 \psi''(x_\theta)}{\psi''(1)} \\
        &= -2 \beta x_{SI} x_{SS} \frac{\psi''(x_\theta)}{\psi'(x_\theta)^2 \psi'(1)} \\
        &= \frac{\dot{x}_{SS}}{\psi'(1)},
    \end{aligned}
\end{equation*}
which shows that~\eqref{eq:VOrigMSS} and the equation for $x_{SS}$ in~\eqref{eq:GR} are equivalent.
Since the remaining equations in both systems rely on the first three equations that we have just shown to be equivalent, the Volz and DSA models are equivalent under their respective initial conditions.
\end{proof}

\bibliographystyle{plain}

\providecommand{\noopsort}[1]{}

\end{document}